\documentclass[letterpaper, 10 pt, conference]{ieeeconf}  

\IEEEoverridecommandlockouts                              
\usepackage{graphics} 
\usepackage{epsfig} 
\usepackage{times} 
\usepackage{amsmath} 
\usepackage{amssymb}  
\usepackage{amsfonts}
\usepackage{flushend} 
\usepackage{xcolor}
\usepackage{url}

\usepackage{amsthm}
\newtheorem{thm}{Theorem}

\theoremstyle{remark}
\newtheorem{defn}{Definition}

\newtheorem{remark}{Remark}

\newcommand{\R}{\mathbb{R}}

\renewcommand{\thefigure}{\arabic{figure}}

\title{\LARGE \bf Learning Stable and Passive Neural Differential Equations
}

\author{Jing Cheng, Ruigang Wang and Ian R. Manchester
\thanks{*This work was supported in part by the Australian Research Council, and the NSW Defence Innovation Network.}
\thanks{The authors are with the Australian Centre for Robotics (ACFR), and the School of Aerospace, Mechanical and Mechatronic Engineering, The University of Sydney, Sydney, NSW 2006, Australia
        {\tt\small ruigang.wang@sydney.edu.au}}%
}

\begin{document}

\maketitle
\thispagestyle{empty}
\pagestyle{empty}

\begin{abstract}

In this paper, we introduce a novel class of neural differential equation, which are intrinsically Lyapunov stable, exponentially stable or passive. We take a recently proposed Polyak Lojasiewicz network (PLNet) as an Lyapunov function and then parameterize the vector field as the descent directions of the Lyapunov function. The resulting models have a same structure as the general Hamiltonian dynamics, where the Hamiltonian is lower- and upper-bounded by quadratic functions. Moreover, it is also positive definite w.r.t. either a known or learnable equilibrium. We illustrate the effectiveness of the proposed model on a damped double pendulum system.



\end{abstract}

\section{Introduction}

Neural networks have demonstrated great power in machine learning and control, particularly on learning the dynamics and forecasting the behavior of dynamical systems \cite{hornik1989multilayer, chen2018neural}. When approximating the dynamic behaviour, preserving some essential properties especially stability and passivity, is favored by learning and control community. Enforcing stability can benefit the learnt models especially in terms of generalization. 

For nonlinear systems, enforcing stability during learning has been studied with Gaussian mixture models and polynomial models in \cite{tobenkin2010convex, khansari2011learning, tobenkin2017convex} and even in the case of linear systems it is non-trivial \cite{umenberger2018maximum}. For nonlinear systems, there are various of stability notions with different implications.  In the context of learning, a strong stability notion called \emph{contraction} \cite{lohmiller1998contraction} (any pair of trajectories converge to each other) has recently received much attention due to its equilibrium-independent stability nature. For discrete-time setup,  \cite{revay2020contracting,revay2023recurrent,fan2024learning} have developed contracting, incrementally passive and dissipative neural dynamics. A continuous-time counterpart can be found in \cite{martinelli2023unconstrained}. A benefit of \cite{revay2023recurrent, martinelli2023unconstrained} is their direct (i.e. unconstrainted \& surjective) parameterization of stable models, making training easy. However a limitation is that they enforce contraction with respect to a state-independent quadratic metric, limiting flexibility. 

For learning dynamical systems with weak stability properties (e.g., Lyapunov stability w.r.t. a particular equilibrium), it is often desired to apply models that preserve the similar stability property. A key component for stable neural differential equations is the neural Lyapunov function. From \cite{wilson1967structure} and the resolution of the Poincare Conjecture by Perelman \cite{anderson2004geometrization}, all Lyapunov functions have level sets homeomorphic to the unit ball. This suggests searching for candidate Lyapunov functions of the form $V(x)=g(x)^\top g(x)$ where $g$ is a neural network, e.g. multi-layer perception (MLP). But for general MLP, $V(x)$ may not be positive definite or radially unbounded. In \cite{richards2018lyapunov}, a special parameterization for $g$ (i.e., each layer has a larger dimension than the previous layer) is proposed in \cite{richards2018lyapunov} such that $V(x)$ is positive definite. The drawback is that it may restrict the model flexibility. The work \cite{kolter2019learning} proposed several neural Lyapunov functions built on the input convex neural network (ICNN) \cite{amos2017input}. To be specific, given an ICNN $\phi(x)$ where $\phi$ is a scalar-output network which is convex w.r.t. the input $x$, the Lyapunov function is constructed as $V(x)=\phi(x)$ or $V(x)=\phi(x)+\epsilon |x|^2$. Although it can guarantee a lower quadratic bound, it is still unknown how to impose an upper quadratic bound. A more recent survey on neural Lyapunov functions can be found in \cite{dawson2023safe}.

Given an Lyapunov function $V(x)$, one needs to construct a vector field $f(x)$ such that the Lyapunov inequality $\nabla^\top V(x)f(x)<0 $ holds for nonzero $x$. An intuitive approach is the gradient flow, i.e., $f(x) = -\alpha(x) \nabla V(x)$ with $\alpha(x)>0$. However, it is often restrictive as there are infinitely many of descent directions. In \cite{kolter2019learning}, a projection method is applied to resolve the case where Lyapunov inequality is infeasible. 

In this work, we take a recently developed bi-Lipschitz network $g$ from \cite{wang2024monotone} to construct a Lyapunov function $H(x)=0.5|g(x)|^2$. Since $g$ is bi-Lipschitz, then it is also invertible, implying that the level set of $H$ is homeomorphic to a unit ball. Unlike the ICNN-based Lyapunov function, the level sets of $H$ can be non-convex. Thanks to the bi-Lipschitz property of $g$, both quadratic lower and upper bounds are guaranteed automatically. We then construct the vector field as $f(x)=[J(x)-R(x)]\nabla H(x)$ where $J$ is skew symmetric and $R$ is positive definite. The above $f$ is a smooth field of the descent vectors of $H$. By parameterizing $J$ and $R$ using neural networks, one can learn flexible stable dynamics. Moreover, this structure can also be realized as the Hamiltonian form that widely exists in physical systems \cite{van2014port,duindam2009modeling}. 

{\bf Contribution.} In this paper, we propose a new class of stable neural differential equations, called stable Hamiltonian neural dynamics (SHND). 
\begin{itemize}
    \item It has certified global Lyapunov stability and exponential stability by the parameterized Lyapunov function, with respect to an equilibrium point. The equilibrium point can either be imposed as a prior knowledge or learned by the data. 
    \item We further extend it to passive port-Hamiltonian system, which has potential applications in learning passivity-based controllers.
    \item Empirical results show that it outperforms existing methods in terms of fitting error, simulation error and robustness.  
\end{itemize}

{\bf Paper organization.} In section~\ref{sec:problem setup} we give a problem definition and review some preliminaries results. We present our main result in section~\ref{sec:hnn}. Section~\ref{sec:exp} illustrates our approach on a damped double pendulum, followed by conclusions in Section~\ref{sec:conclusion}. 

{\bf Notation.} Given a $x\in\R^n$, we denote a ball around $x$ with radium $\epsilon$ by  $\mathcal{B}_\epsilon(x):=\{y\in \R^n \mid |y-x|\leq \epsilon \}$ where $|\cdot|$ is the Euclidean norm. 
We denote $A\succeq B$ if and only if the matrix $(A-B)$ is positive semi-definite. 
We define the gradient transpose operator as $\nabla_x:=(\partial /\partial x)^\top$. The subscript will be omitted when it is clear from the context. 

\section{Problem Formulation and Preliminaries}
\label{sec:problem setup}
\subsection{Problem setup}

Given a dataset $\mathcal{D}=\{(x_i,v_i)\mid x_i,v_i\in \R^n\}$, we consider the problem of learning a neural vector field $v=f(x)$ from $\mathcal{D}$ such that the following differential equation
\begin{equation}\label{eq:system}
    \dot{x}=f(x)
\end{equation} 
is \emph{guaranteed} to be stable in the following senses.  
\begin{defn}
    Let $x^\star$ be an equilibrium of \eqref{eq:system}, i.e., $f(x^\star)=0$. 
    \begin{enumerate}
        \item System \eqref{eq:system} is said to be \emph{stable} at $x^\star$ if, for each $\epsilon>0$, there is $\delta=\delta(\epsilon)>0$ such that 
    \begin{equation}\label{eq:stable}
        x(0)\in \mathcal{B}_{\delta}(x^\star)\Rightarrow x(t)\in \mathcal{B}_{\epsilon}(x^\star), \quad \forall t\geq 0.
    \end{equation}
    \item It is said to be \emph{exponentially stable} if there exist $\delta, \lambda, \kappa > 0$ such that 
    \begin{equation}\label{eq:exp-stable}
        x(0)\in \mathcal{B}_{\delta}(x^\star)\Rightarrow |x(t)-x^\star|\leq \kappa e^{-\lambda t} |x(0)-x^\star|
    \end{equation}
    for all $t\geq 0$. It is said to be \emph{globally} exponentially stable if \eqref{eq:exp-stable} holds for all $\delta>0$.
    \end{enumerate}
\end{defn}

We also consider learning $v=f(x,u)$ and $y=h(x)$ from a dataset $\mathcal{D}=\{(x_i,v_i, u_i, y_i)\mid x_i,v_i\in \R^{n}, u_i,y_i\in \R^{m}\}$ such that the following controlled system in the input-state-output form
\begin{equation}\label{eq:open-system}
    \dot{x}=f(x,u),\quad y=h(x)
\end{equation}
is internally stable and passive. 
\begin{defn}[\cite{khalil2002nonlinear}]
    System \eqref{eq:open-system} is said to be \emph{passive} if there exists a continuously differentiable positive semidefinite function $S(x)$ (storage function) such that 
    \begin{equation}
        \frac{\partial S}{\partial x} f(x,u) \leq u^\top y, \quad \forall (x,u)\in \R^n\times \R^m.
    \end{equation}
\end{defn}

\subsection{Preliminaries}\label{sec:plnet}
We review two types of neural networks which will be applied in our stable model construction.
\begin{defn}
    A neural network $g:\R^n\rightarrow\R^n$ is said to be \emph{bi-Lipschitz} if there exist $\nu \geq \mu >0$ such that 
    \begin{equation}
        \mu|x-x'|\leq |g(x)-g(x')|\leq \nu |x-x'|,\quad \forall x, x'\in \R^n.
    \end{equation}
\end{defn}
We call such $g$ a $(\mu,\nu)$-Lipschitz network. It is worth to point out that $g$ is also invertible and $g^{-1}$ is a bi-Lipschitz network with bound of $ (1/\nu, 1/\mu) $. In literature, there are some existing approaches \cite{chen2019residual,behrmann2019invertible,lu2021implicit,ahn2022invertible} to learn bi-Lipschitz neural networks with certified bounds. 

The following definition from \cite{wang2024monotone} was motivated by Polyak-\L{}ojasiewicz condition \cite{polyak1963gradient,lojasiewicz1963topological} 
\begin{defn}
    A scalar-output neural network $H: \R^n\rightarrow \R$ is called a \emph{Polyak-\L{}ojasiewicz network} (PLNet) if there exists a  $m>0$ such that
    \begin{equation}\label{eq:pl-cond}
        \left|\nabla_x H\right|^2\geq 2m \bigl(H(x)-H^\star\bigr),\quad \forall x\in \R^n,
    \end{equation}
    where $H^\star:=\min_{y\in \R^n} H(y)$.
\end{defn}

One main result of \cite{wang2024monotone} is that given any $(\mu,\nu)$-Lipschitz network $g$, we can obtain a PLNet as follows:
\begin{equation}\label{eq:Hx}
    H(x)=0.5|g(x)|^2+c,\quad c\in \R,
\end{equation}
i.e., \eqref{eq:pl-cond} holds with $m=\mu^2$. Since our model construction will only use the gradient $\nabla H$, we assume $c=0$ throughout this paper. Some nice properties of PLNet include: 
\begin{enumerate}
    \item $H$ could be non-convex but always has one minimum $x^\star=g^{-1}(0)$, i.e., 
    \begin{equation}\label{eq:H-grad}
        \nabla H(x^\star)=0,\quad \nabla H(x)\neq 0,\; \forall x\in \R^n\backslash\{x^\star\}.
    \end{equation}
    \item $H$ is upper- and lower-bounded by  quadratic functions: 
    \begin{equation}\label{eq:quadratic-bound}
        0.5\mu^2 |x-x^\star|^2 \leq H(x) \leq 0.5 \nu^2 |x-x^\star|^2.
    \end{equation}
    \item If $g$ is a smooth mapping, then the level set $\mathbb{L}_\alpha=\{x\in \R^n\mid H(x)\leq \alpha\}$ is homeomorphic to a unit ball in $\R^n$, for any $\alpha>0$. 
\end{enumerate}
These features make the PLNet a promising candidate for learning Lyapunov functions, e.g. using the methods of \cite{pmlr-v155-boffi21a, dawson2023safe}. In this paper we focus on learning stable dynamics by parameterising flows that are descent directions for such a Lyapunov function, and consider further extensions to passive dynamics.

\section{Stable Hamiltonian Neural Dynamics}
\label{sec:hnn}
In this section, we first present a class of stable neural dynamics in the form of Hamiltonian system, and then extend it to passive dynamical models.

\subsection{Stable Hamiltonian models}
Our proposed model has the form of general Hamiltonian dynamics as follows:  
\begin{equation}\label{eq:HNN}
    \dot x= [J(x)-R(x)]\nabla H(x)
\end{equation}
where $x(t)\in \R^{n}$ is the state, $H(x)\in \R$ is the Hamiltonian, and $J(x)=-J^\top(x)$, $R(x)=R^\top(x)$ are the interconnection and damping matrices, respectively, which will be parameterized by neural networks. We now state our main result.
\begin{thm}\label{thm:1}
    Suppose that $H$ is a PLNet \eqref{eq:Hx} with $g$ as a $(\mu,\nu)$-Lipschitz network. Let $x^\star$  be the global minimum of $H$. We have the following results:
    \begin{enumerate}
        \item[a)] If $R(x)\succeq 0,\, \forall x\in \R^n$, then \eqref{eq:HNN} is stable at $x^\star$, i.e., \eqref{eq:stable} holds with $\delta(\epsilon)=\epsilon \mu/\nu$.
        \item[b)] If there exist $\epsilon > 0$ such that $R(x)\succeq \epsilon I,\,\forall x\in \R^n $, then \eqref{eq:HNN} is globally exponentially stable at $x^\star$, i.e., \eqref{eq:exp-stable} holds with $\kappa=\nu/\mu$ and $\lambda=\epsilon$.
    \end{enumerate}
\end{thm}
\begin{proof}
a) As shown in Section~\ref{sec:plnet}, $H(x)$ is a positive semidefinite function and its derivative satisfies
\begin{equation}\label{eq:dotH}
    \begin{split}
        \dot H &= \nabla^\top H(x) [J(x)-R(x)]\nabla H(x) \\
        &= -\nabla^\top H(x) R(x)\nabla H(x) \leq 0,\quad \forall x \in \R^n
    \end{split}
\end{equation}  
implying $H(x(t))\leq H(x(0))$. From \eqref{eq:quadratic-bound}, we have 
\begin{equation*}
    |x(t)-x^\star|\leq \frac{\sqrt{2H(x(t))}}{\mu}\leq \frac{\sqrt{2H(x(0))}}{\mu}\leq \frac{\nu}{\mu}|x(0)-x^\star|,
\end{equation*}
implying $|x(0)-x^\star|\leq \mu\epsilon/\nu\Rightarrow |x(t)-x^\star|\leq \epsilon$ for all $t\geq 0$.

b) Since $R(x)\succeq \epsilon I$, then \eqref{eq:dotH} implies
\begin{equation}
    \dot H \leq - \epsilon |\nabla H|^2 \leq -2\epsilon\mu^2 (H(x)-H(x^\star))=-2\epsilon\mu^2 H(x),
\end{equation}
Furthermore, we have 
\begin{equation}
    \begin{split}
        |x(t)-x^\star|&\leq \frac{\sqrt{2H(x(t))}}{\nu}\leq \frac{\sqrt{2e^{-2\lambda t}H(x(0))}}{\nu} \\
        &\leq \frac{\nu}{\mu}e^{-\lambda t}|x(0)-x^\star|,
    \end{split}
\end{equation}
where $\lambda=\epsilon \mu^2$, i.e., \eqref{eq:HNN} is globally exponentially stable with rate $\epsilon\mu^2$ and overshoot $\nu/\mu$.
\end{proof}
\begin{remark}
    The hyper-parameters $\mu,\nu$ and $\epsilon$ offer us the flexibility to control model's minimum convergence rate and overshoot.
\end{remark}
\begin{remark}
    The Hamiltonian $H(x)$ is a Lyapunov function $V(x)$ of system \eqref{eq:HNN}. 
\end{remark}

If the equilibrium $x^\star$ is unknown, we simply take the PLNet \eqref{eq:Hx} as the learnable Hamiltonian. Otherwise, we can impose the prior knowledge of $x^\star$ into our model by choosing the Hamiltonian $H(x)=0.5|g(x)-g(x^\star)|^2$, which also is a PLNet as $\tilde{g}(x):=g(x)-g(x^\star)$ is a bi-Lipschitz model with the same bound as $g$. In this work we will use the bi-Lipschitz network recently proposed in \cite{wang2024monotone}. Specifically, $y=g_{\theta}(x)$ has the form of 
\begin{equation}\label{eq:network}
    \begin{split}
        z_k&=\sigma(W_k z_{k-1}+ U_{k}x+b_k),\; z_0=0\\
        y &= \mu x + \sum_{k=1}^L Y_k z_k +b_y
    \end{split}
\end{equation}
where $z_k$ is the $k$th layer hidden unit, and $x,y$ are the input and output of this network. The learnable parameters of $g$ are $\theta=\{U_k, Y_k, W_k, b_k, b_y\}$. One main result of \cite{wang2024monotone} is a model parameterization $\mathcal{M}:\phi\rightarrow\theta$ with $\phi\in \R^N$ as a free variable such that for any $\phi\in \R^N$, the model $g_{\mathcal{M}(\phi)}$ is bi-Lipschitz for some prescribed bound of $(\mu, \nu)$. Then, the off-shelf optimization algorithm (e.g., stochastic gradient decent) can be directly applied to train bi-Lipschitz models. 

We now parameterize the matrix functions $J(x)$ and $R(x)$. Let $S(x) $ and $L(x)$ be two learnable multi-layer perceptions (MLPs), which are mappings from $\R^n$ to $\R^{n\times n}$. We then construct
\begin{equation}\label{eq:JR}
    \begin{split}
        J(x)=S(x)-S^\top(x),\quad 
        R(x)=L^\top(x)L(x)+\epsilon I
    \end{split}
\end{equation}
where $\epsilon\geq 0$ is a user-specified constant. 

\subsection{Passive port-Hamiltonian models}
We further extend \eqref{eq:HNN} to port-Hamiltonian system of the form  
\begin{equation}\label{eq:passive-model}
    \begin{split}
        \dot x &= (J(x)- R(x))\nabla H(x) + B(x)u \\
        y &= B(x)^\top \nabla H(x)
    \end{split}
\end{equation}
where $u(t),y(t)\in \R^m$ are the input-output variables, respectively, and $B(x)\in \R^{n\times m}$ is a learnable matrix function. The following theorem shows that \eqref{eq:passive-model} is passive with respect to the storage function $H(x)$. 
\begin{thm}
    Suppose that $H$ is a PLNet and the matrix function $R$ satisfies $R(x)\succeq 0$ for all $x$. Then, \eqref{eq:passive-model} is passive.
\end{thm}
\begin{proof}
    We take $H(x)$ as the storage function and then check the dissipation inequality as follows:
    \begin{equation}
        \begin{split}
            \dot H-u^\top y&=-\nabla^\top H R \nabla H+\nabla^\top H B u-u^\top B^\top \nabla H \\
            &=-\nabla^\top H R \nabla H \leq 0.
        \end{split}
    \end{equation}
\end{proof}
\begin{remark}
    An potential application of the above model class is learning stabilizing controllers for unknown but passive systems \cite{secchi2007control}. 
\end{remark}

\section{Experiments}
\label{sec:exp}
We illustrate the proposed approach on learning the dynamics of a damped double pendulum system. The pendulum state is $x=[\,\theta_1 \; \theta_2 \; \dot \theta_1 \; \dot \theta_2\,]^\top $ where $\theta_1, \theta_2$ are the pendulum joint angles and $\dot{\theta}_1,\dot{\theta}_2$ are the angular velocities. The dynamics of the damped pendulum is $\dot x = f_p(x)$
where the vector filed $f_p:\R^4\rightarrow \R^4$ can be found in \cite{kolter2019learning}. Due to damping on each joint, the above system has a stable equilibrium $x^\star=0$, i.e., $f_p(0)=0$. Code is available at \url{https://github.com/ruigangwang7/StableNODE}.

The main objective is to learn a vector field $f:\R^4\rightarrow \R^4$ based on a set of samples $(x_i, v_i)$ where $v_i=f_p(x_i)$. To be specific, we randomly generate 2000 samples as the train data set $\mathcal{D}_{\mathrm{train}}:=\{(x_i, v_i)\}_{1\leq i\leq 2000}$, where each $x_i$ is uniformly drawn from the region where $|\theta_i|\leq \pi$ $\mathrm{rad}$ and $|\dot{\theta}_i|\leq \pi$ $ \mathrm{rad}/\mathrm{s}$ with $i=1,2$. The test data set $\mathcal{D}_{\mathrm{test}}$ consists of 500 samples with the same distribution. Given a learnable model $f$, the train/test loss are defined as the mean square $\ell_2$ loss:
\begin{equation}\label{eq:l2loss}
    \ell(f)=\frac{1}{n_{\mathcal{D}}}\sum_{i=1}^{n_{\mathcal{D}}} \bigl|v_i-f(x_i)\bigr|^2.
\end{equation}
with $n_{\mathcal{D}}$ as the number of elements in $\mathcal{D}$, where $\mathcal{D}$ is either $\mathcal{D}_{\mathrm{train}}$ or $\mathcal{D}_{\mathrm{test}}$. 

Beside loss minimization, we also want to learn a model $f$ such that the associated dynamics is globally exponentially stable at $x^\star$. The main purpose of such requirement is that when simulating the model $f$, a small fitting loss \eqref{eq:l2loss} can lead to a small error between the trajectories generated by $\dot x= f_p(x)$ and $\dot x=f(x)$.

\subsection{Model architecture and training details}
\begin{figure}[!bt]
    \includegraphics[width=\linewidth]{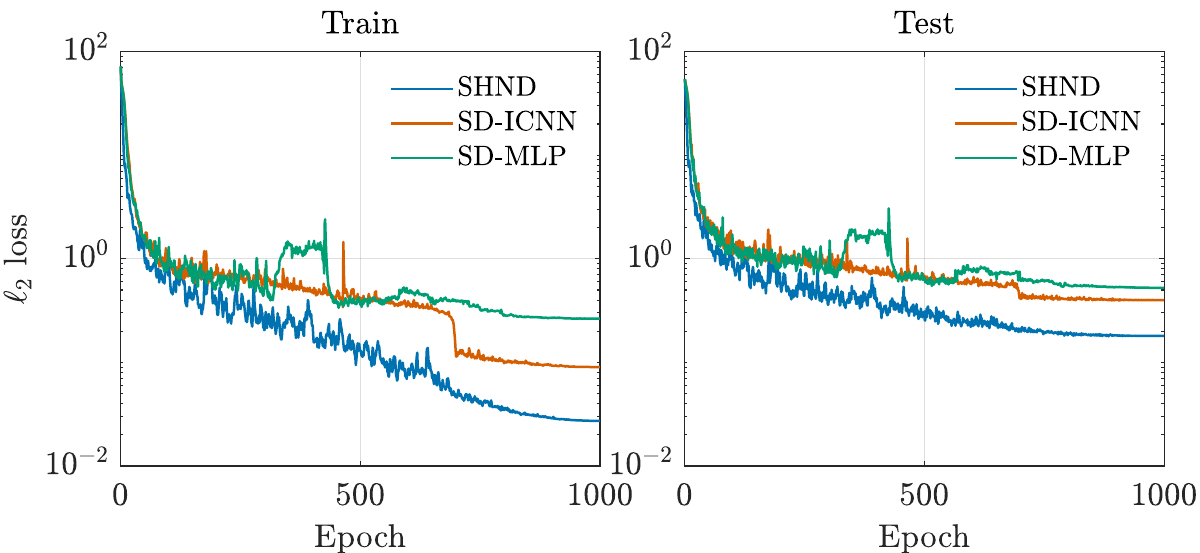}
    \caption{Training and test loss \eqref{eq:l2loss} versus epochs. The proposed SHND model achieves better performance than SD-ICNN and SD-MLP.}\label{fig:loss}
\end{figure}

We will compare our model with the Lyapunov projection based stable models from \cite{kolter2019learning}, which learn a pair of vector field $\hat f:\R^n\rightarrow \R^n$ and an Lyapunov function $V:\R^n\rightarrow \R$. The dynamical model is defined by
\begin{equation}\label{eq:f_proj}
\begin{aligned}
\dot x=f(x)
= \hat{f}(x) - \nabla V(x)\frac{\texttt{ReLU}(\nabla V(x)^T \hat{f}(x) + \alpha V(x))}{| \nabla V(x) |^2}
\end{aligned}
\end{equation}
where $\texttt{ReLU}(x)=\max(0,x)$. For the vector field $\hat f$, we choose a 4-100-100-4 multi-layer perception (MLP), i.e., a network with input and output dimension of 4 and two hidden layer with each has 100 neurons. For the Lyapunov function $V$, we consider the following two cases extracted from \cite{kolter2019learning}:  
\begin{itemize}
    \item {\bf SD-MLP}: We take a 4-64-64 MLP $g$ and then construct $V(x)=|g(x)|^2$. And the resulting model in \eqref{eq:f_proj} is referred as SD-MLP. Note that such choice of Lyapunov function generally does not satisfy \eqref{eq:H-grad} and \eqref{eq:quadratic-bound}. Thus, \eqref{eq:f_proj} only ensures that $V(x(t))$ is bounded. The simulation trajectory may contain oscillations.
    \item {\bf SD-ICNN}: We take a 4-64-64-1 ICNN \cite{amos2017input} as the Lyapunov function $V(x)$. Although $V(x)$ is convex, it may have multiple global minimum. By adding a quadratic term $\epsilon |x|^2$ to $V$, one can ensure that the resulting Lyapunov function has a quadratic lower bound. But it is still not clear how to enforce the quadratic upper bound. Moreover, when the prior knowledge about the equilibrium $x^\star$ is available, it is unclear how to enforce the constraint $V(x^\star)=0$ to the ICNN. Therefore, the resulting dynamics \eqref{eq:f_proj} may converge to an undesired equilibrium.
\end{itemize}

In summary, since the existing approaches for constructing neural Lyapunov function do not guarantee 
\[
c_1|x|^2\leq V(x) \leq c_2|x|^2
\]
with $0<c_1\leq c_2$, then the dynamics \eqref{eq:f_proj} does not provide exponentially guarantees for the equilibrium $x^\star$.

For the proposed SHND model, we use the bi-Lipschitz network \eqref{eq:network} from \cite{wang2024monotone} with architecture 4-32-32-4 and Lipschitz bound $(\mu,\nu)=(0.1,2)$ to construct the Hamiltonian $H(x)$. For the matrix functions $J(x)$ and $R(x)$, we first take a 4-90-90-32 MLP, then partition the output into two $4\times 4$ matrices $L$ and $S$, and finally compute $J,R$ via \eqref{eq:JR} with $\epsilon=0.01$. As shown in Thm~\ref{thm:1}, the SHND model can provide global exponential stability guarantee for the equilibrium $x^\star$. We will illustrate this point in the next section. For a fair comparison, our SHND model has a similar amount of learnable model parameters (around 15K) as the SD-MLP and SD-ICNN. 

We use Adam \cite{kingma2014adam} to train each model with mean square loss for 1000 epochs. The batch size is 200 and the learning rate starts from $0.01$ and decays to $0$ according to the {\tt cosine} rate scheduler. 

\subsection{Results and discussions}

We first show the training and test loss over the epochs in Fig.~\ref{fig:loss}. Our proposed SHND model outperforms SD-MLP and SD-ICNN. To test the model stability, we simulate the learned dynamics $\dot x=f(x)$ as well as the pendulum dynamics $\dot x = f_p(x)$ and then compare the error between those trajectories. To be specific, we take a batch of 100 samples $|\theta_i|\leq \pi/2$ and $\dot\theta_i=0$ as initial conditions, then simulate the dynamics using Runge-Kutta method with step size of $0.01\mathrm{s}$. We plot the batch averaged and maximum simulation errors in Fig.~\ref{fig:mse}, as well as the trajectories w.r.t. two initial conditions in Fig.~\ref{fig:simu}. 

\begin{figure}[!bt]
    \centering
    \includegraphics[width=0.95\linewidth]{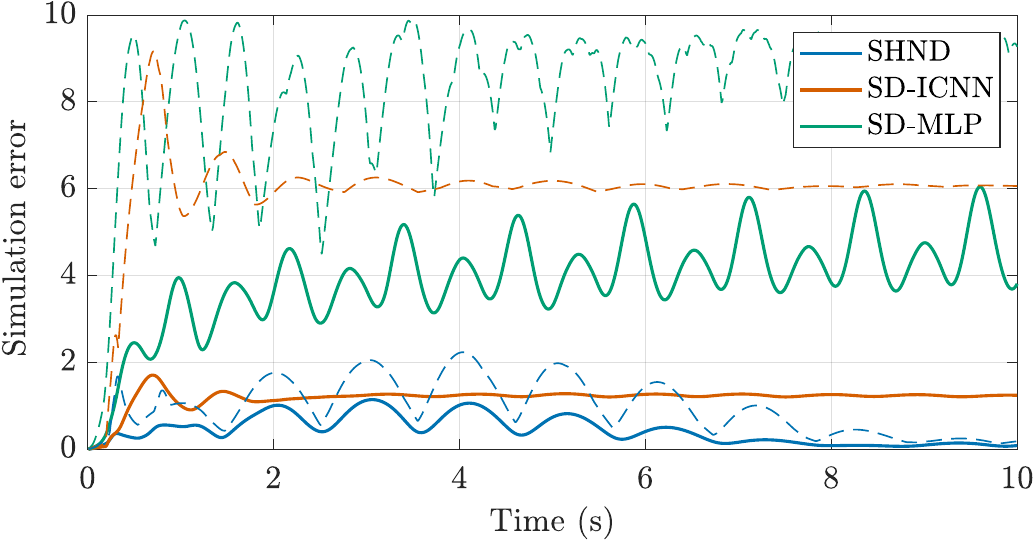}
    \caption{Batch averaged (solid) and maximum (dashed) simulation error between the trajectories generated by the learned model and the original pendulum dynamics w.r.t. a batch of initial conditions. Our model has much smaller simulation error and converges to 0 as time increases. This is mainly due to the stability guarantee w.r.t. the known equilibrium $x^\star$. The SD-ICCN has a steady error due to the fact that it converges to another equilibrium for some initial conditions. The SD-MLP produces oscillating trajectories. }\label{fig:mse}
\end{figure}

We can observe that our model has much smaller simulation errors compared to SD-ICNN and SD-MLP. In particular, the simulation error of SHND converges to 0 as the simulation time increases, demonstrating that our model can provide stability guarantee for the known equilibrium $x^\star$. For SD-ICNN model, as shown in Fig.~\ref{fig:simu}, with slight changes in the initial condition, it converges to a different equilibrium. The SD-MLP produces oscillating behaviors as its Lyapunov function is not positive definite. Also, its simulation error is much larger than the other two approaches.

\begin{figure}[!bt]
    \includegraphics[width=\linewidth]{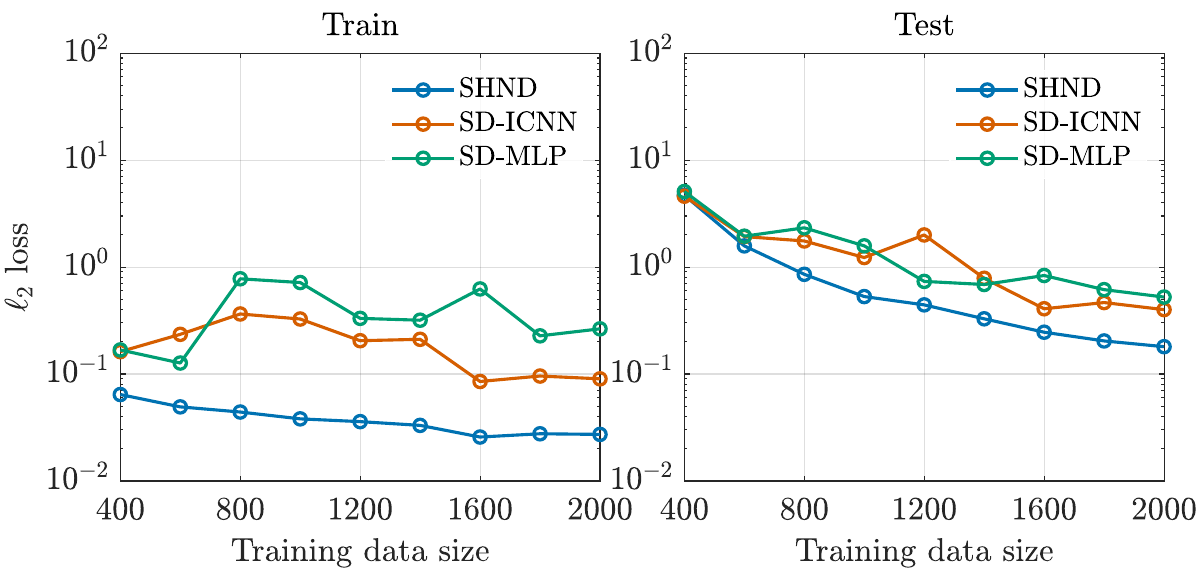}
    \caption{Training and test loss \eqref{eq:l2loss} versus training data size. We use the same test data set with 500 samples. For our model, the performance improves as the size increases while the other two have some variations. }\label{fig:loss_v_tsize}
\end{figure}

\begin{figure}[!bt]
    \centering
    \includegraphics[width=\linewidth]{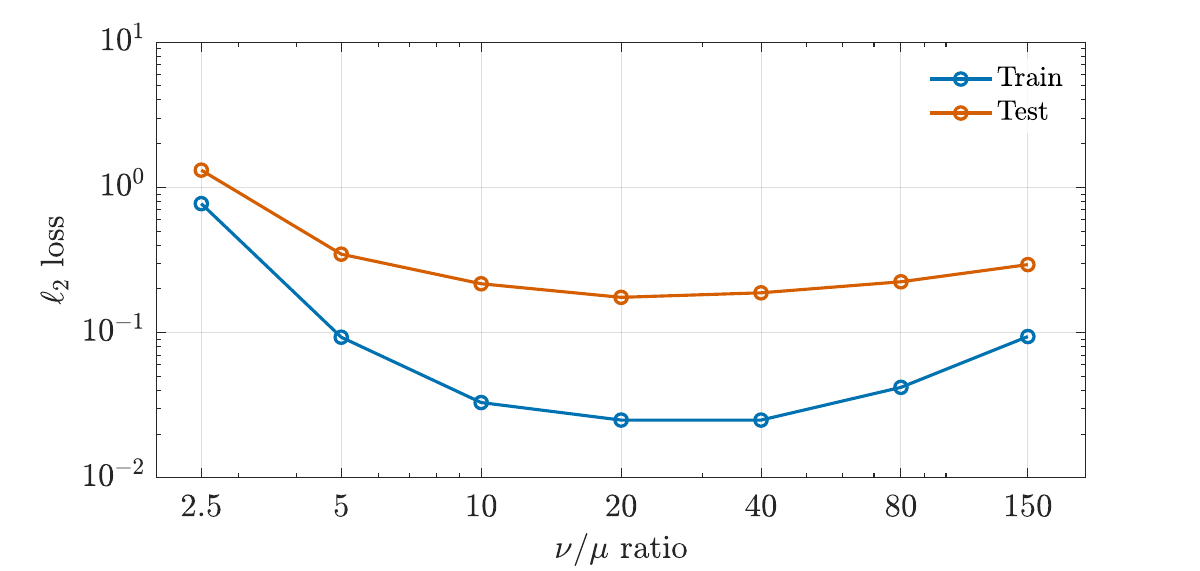}
    \caption{Effect of the ratio $\nu/\mu$ on the train/test loss \eqref{eq:l2loss}. We set $\mu = 0.1$ and trained the SHND models with different $\nu$. The ratio $\nu/\mu$ can be serve as a regularizer for the proposed SHND model. } \label{fig:mu_n_nu}
\end{figure}

\begin{figure*}[!bt]
    \centering
    \includegraphics[width=\linewidth]{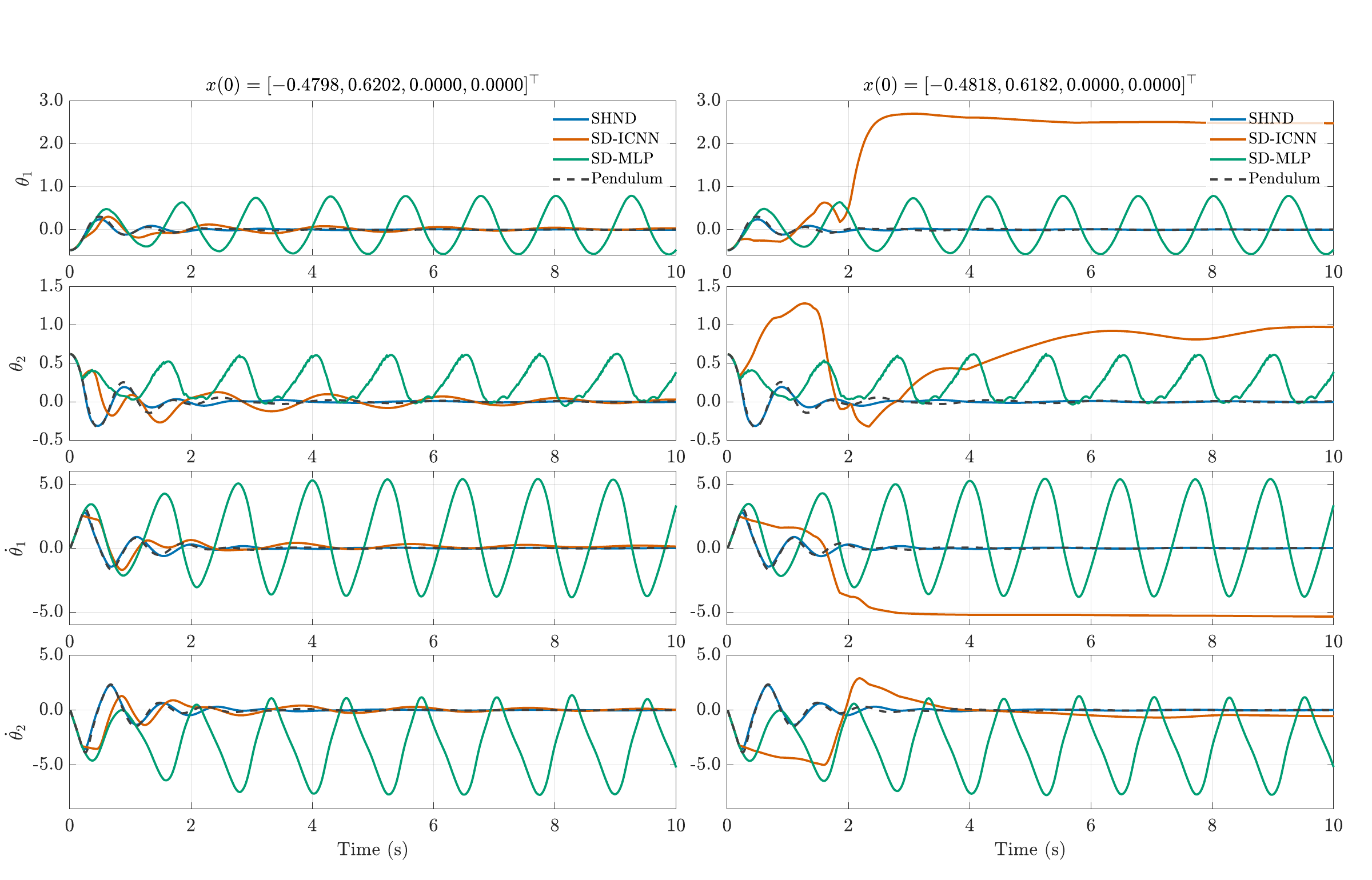}
    \caption{Simulation trajectories of different models with two initial states. The SD-MLP model generates oscillating behaviors. Both SD-ICNN and the proposed SHND converge. But the SD-ICNN converges to a non-zero equilibrium point despite a small change in the initial condition. }\label{fig:simu}
\end{figure*}

We also perform additional model training with respect to different data sizes. From Fig.~\ref{fig:loss_v_tsize} we can see that our model has a constant performance improvement as the data size increases while the other two approaches have some variation. 

Different hyper-parameter $\nu/\mu$ in our model is also tested. The result in Fig.~\ref{fig:mu_n_nu} shows that $\nu/\mu$ has a regularization effect on our model.


\section{Conclusion}
\label{sec:conclusion}
In this paper, we have proposed a new class of neural differential equations called stable Hamiltonian neural dynamics (SHND), which has built-in stability \textit{w.r.t.} an equilibrium. Empirical results on a vector field fitting problem demonstrate its effectiveness. For future directions, we will consider to learn stable models from time-series data and also explore applications in learning passivity-based controllers.

\bibliographystyle{ieeetr}
\bibliography{ref}

\end{document}